\documentclass[11pt]{article}
\usepackage{ifthen}
\usepackage{microtype}
\usepackage{boxedminipage}
\usepackage{amsmath}
\usepackage{amsfonts}
\usepackage{amssymb}
\usepackage{float}
\usepackage[small,bf]{caption}
\usepackage{authblk}
\usepackage{amssymb,amsfonts,amsthm,epsfig}
\newtheorem{theorem}{Theorem}

\newtheorem{lemma}{Lemma}

\newtheorem{corollary}{Corollary}

\newtheorem{definition}{Definition}

\newfloat{Algorithm}{hbtp}{lop}[section]

\newcommand{\pr}{{\bf Pr}}
\newcommand{\eps}{{\epsilon}}
\newcommand{\C}[1]{{{C}^{(#1)}}}
\newcommand{\itS}[1]{{{S}^{(#1)}}}
\newcommand{\itT}[1]{{{T}^{(#1)}}}
\newcommand{\etal}{{\it et al.}}

\begin{document}

\title{Improved analysis of $D^2$-sampling based PTAS for $k$-means and other Clustering problems}

\author[1]{Ragesh Jaiswal \thanks{Corresponding author: \texttt{rjaiswal@cse.iitd.ac.in}}}
\author[1]{Mehul Kumar}
\author[1]{Pulkit Yadav}

\affil[1]{Department of Computer Science and Engineering, IIT Delhi}

\date{}

\maketitle

\begin{abstract}
We give an improved analysis of the simple $D^2$-sampling based PTAS for the $k$-means clustering problem given by Jaiswal, Kumar, and Sen~\cite{jks12}.
The improvement on the running time is from $O\left(nd \cdot 2^{\tilde{O}(k^2/\eps)}\right)$ to $O\left(nd \cdot 2^{\tilde{O}(k/\eps)}\right)$.
\end{abstract}

\section{Introduction}
We present this work as a short note and this may be regarded as an extension of the work by Jaiswal, Kumar, and Sen~\cite{jks12}.
We would like to point the reader to \cite{jks12} for more elaborate discussions regarding the past work in designing PTAS for the $k$-means problem and sampling based algorithms.
Here, we will only discuss our improved analysis.
However, we will not skip any essential details about the algorithm or the analysis and this paper will be self-contained.
Note that this will be at the cost of having some basic material borrowed from \cite{jks12}.

\paragraph{Generalized $k$-median problem and $k$-means} 
The general $k$-median problem may be defined as follows: 
Let $\mathcal{X}$ be any input space and let $D$ denote the distance measure defined over pair of points of this input space. 
Given a set $P \subseteq \mathcal{X}$ with $|P| = n$, and a parameter $k$, find a set $C \subseteq \mathcal{X}$ with $|C| = k$ such that the following objective function is minimized: $\Delta(P, C) = \sum_{p \in P} D(p, C)$.
Here, $D(p, C)$ denotes the distance of $p$ from the closest center in $C$.
The $k$-means problem is a special case of the generalized $k$-median problem where $\mathcal{X} = \mathbb{R}^d$ and $\forall p, q \in \mathcal{X}, D(p, q) = e(p, q)^2$, where $e(p, q)$ denotes the Euclidean distance between points $p$ and $q$.
We will discuss both the generalized problem and the $k$-means problem.
We give the algorithm and analysis for the generalized problem and then state the result related to the $k$-means problem.

\paragraph{Related work}
The $k$-means problem is $\mathsf{NP}$-hard and there are exact algorithms that solve the problem in time $\Omega(n^{kd})$.
Moreover, there are hardness of approximation results known for this problem showing that approximating beyond a fixed constant factor is unlikely. 
However, there has been significant research in obtaining $(1 + \eps)$-approximation algorithms for the problem that run in time that is exponential in $k$ since in many cases $k$ can be assumed to be a small constant.
We usually call such algorithms PTAS (under the implicit assumption that $k$ is a constant).
We advise the reader to see a more detailed description of the previous work in \cite{jks12}.
The running time of our algorithm match those of Ackermann \etal~\cite{abs10}.
They use an improved analysis of a PTAS by Kumar \etal~\cite{KumarSS10} and show that the running time of their PTAS is $O(n d \cdot 2^{\tilde{O}(k/\epsilon)})$.
\footnote{$\tilde{O}$ notation hides a $O(\log{(k/\eps)})$ factor which simplifies the expression.}

\paragraph{Our contributions} 
In this paper, we give an improved analysis of the PTAS of Jaiswal \etal~\cite{jks12} and using this obtain a PTAS for the $k$-means problem with running time $O(n d \cdot 2^{\tilde{O}(k/\epsilon)})$. 
The analysis in~\cite{jks12} crucially used a {\em separation} or {\em irreducibility} property. 
The main idea of the new analysis is to remove this dependency.
\footnote{The reader familiar with \cite{jks12} may note that the separation condition was used in \cite{jks12} to show that while sampling, there is an optimal cluster such that {\em all} points from that cluster have a good chance of being sampled.
The crucial observation in this work is that this property is not really required since the points that have small chance to be sampled are close to the already chosen centers. 
So, we just need to worry about picking a good center for the points that have a good chance of being sampled.}
Using the coreset construction of Chen~\cite{Chen09} with our algorithm, we can obtain a $(1+\eps)$-approximation algorithm with running time $O\left(nkd + d^2 \cdot 2^{\tilde{O}(k/\eps)} \cdot \log^{k+2}{(n)} \right)$.
This is a standard method to obtain a faster algorithm.
The reader is encouraged to see the work of Chen~\cite{Chen09} and Ackermann and Bl\"{o}mer~\cite{ab10} for this.

\paragraph{Organization} We give some preliminary results in Section~\ref{sec:pre}. We give the description of our algorithm and its analysis in Section~\ref{sec:algo-analysis}. In Section~\ref{sec:k-means}, we state our results for the $k$-means problem.

\section{Preliminaries}
\label{sec:pre}

Given a set of points $P \subseteq \mathcal{X}$ and a set of centers $C \subseteq \mathcal{X}$, $\Delta(P, C) = \sum_{p \in P} \min_{c \in C} D(p, c)$.
For singleton $C = \{c\}$, we shall often abuse notation and use $\Delta(P, c)$.
For any positive integer $k$, we denote the cost of the optimal $k$-means solution for $P$ using $\Delta_k(P)$.  

Our results hold, given that the dissimilarity measure satisfies certain properties. 
Note that many of the measures that are actually used in practice, do satisfy these properties and hence they are not restrictive. 
The reader should see previous works \cite{jks12, abs10} for detailed discussions.
Here, we state the properties that we will be using.
First, we will describe a property that apart from minor differences, is basically a definition by Ackermann et al. \cite{abs10} (see Theorem 1.1) who discuss PTAS 
for the $k$-median problem with respect to metric and non-metric distance measures.

\begin{definition}[$(f, \gamma, \delta)$-Sampling property]
Given $0 < \gamma, \delta \leq 1$ and $f: \mathbb{R} \times \mathbb{R} \rightarrow \mathbb{R}$, a distance measure $D$ over space $\mathcal{X}$ is said to have $(f, \gamma, \delta)$-sampling property if the following holds:
for any set $P \subseteq \mathcal{X}$, a uniformly random sample $S$ of $f(\gamma, \delta)$ points from $P$ satisfies 
$$\pr\left[\sum_{p \in P} D(p, \Gamma(S)) \leq (1+\gamma) \cdot \Delta_1(P)\right] \geq (1 - \delta),$$
where $\Gamma(S) = \frac{\sum_{s \in S} s}{|S|}$ denotes the mean of points in $S$.
\end{definition}

\begin{definition}[Centroid property]
A distance measure $D$ over space $\mathcal{X}$ is said to satisfy the centroid property if for any subset $P \subseteq \mathcal{X}$ and any point $c \in \mathcal{X}$, we have: $\sum_{p \in P} D(p, c) = \Delta_{1}(P) + |P| \cdot D(\Gamma(P), c)$,
where $\Gamma(P) = \frac{\sum_{p \in P} p}{|P|}$ denotes the mean of the points in $P$.
\end{definition}

\begin{definition}[$\alpha$-approximate triangle inequality]
Given $\alpha \geq 1$, a distance measure $D$ over space $\mathcal{X}$ is said to satisfy $\alpha$-approximate triangle inequality if for any three points $p, q, r \in \mathcal{X},
D(p, q) \leq \alpha \cdot (D(p, r) + D(r, q))$
\end{definition}

\begin{definition}[$\beta$-approximate symmetry]
Given $0 < \beta \leq 1$, a distance measure $D$ over space $\mathcal{X}$ is said to satisfy $\beta$-symmetric property if for any pair of points $p, q \in \mathcal{X}$, $\beta \cdot D(q, p) \leq D(p, q) \leq \frac{1}{\beta} \cdot D(q, p) $
\end{definition}

\noindent
We will also use the Chernoff bound.
\begin{theorem}[Chernoff bound]\label{thm:chernoff}
Let $X_1, ..., X_n$ be independent 0/1 random variables. Let $X = \sum_{i} X_i$ and $\mu = \sum_i E[X_i]$. 
Let $\delta > 0$ be any real number. Then  $\pr[X \leq (1 - \delta) \cdot \mu] \leq e^{-\delta^2 \mu/2}$.
\end{theorem}

\section{Our Algorithm and Analysis}
\label{sec:algo-analysis}

\begin{definition}[$D^2$-sampling]
Given a set of points $P$ and a set of centers $C$, a point $p \in P$ is said to be sampled using {\em $D^2$-sampling}
with respect to $C$ if the probability of it being sampled, $\rho(p)$,  is given by
$\rho(p) = \frac{D(p,C)}{\sum_{x \in P} D(x,C)} = \frac{\Delta(\{p\}, C)}{\Delta(P,C)}.$
\end{definition}

Here is a high level description of our algorithm. 
Essentially, the algorithm maintains a set $C$ of centers, where $|C| \leq k$. 
Initially $C$ is empty, and in each iteration, it adds one center to $C$ till its size reaches $k$. 
Given a set $C$, it samples a multiset $S$ of $N$ points from $P$ using $D^2$-sampling with respect to $C$. 
Then it picks  a subset $T$ of $S$ of size $M$, and adds the centroid of $T$ to $C$. 
The algorithm cycles through all possible $\binom{N}{M}$ subsets of size $M$ of $S$ as choices for $T$, and for each such choice,
repeats the above steps to find the next center, and so on. 
Repetition is done for probability amplification.
Algorithm~\ref{fig:k} gives a concise representation of the above algorithm.
Next, we state our main theorem with respect to our algorithm and prove this theorem in the rest of this section.

\begin{center}
\begin{Algorithm}[h]
\begin{boxedminipage}{6.5in}
{\bf Find-$k$-means($P, k, \xi$)}

\hspace{0.1in} - Let $\eta = \frac{2 \alpha^2}{\beta^2} \left(1+\frac{1}{\beta} \right)$, $N = \frac{64 \alpha \eta k}{\beta \eps^2} \cdot f \left( \frac{\eps}{2 \eta}, 0.2 \right)$, 
%
$M = f \left( \frac{\eps}{2 \eta}, 0.2 \right) $, $\kappa = \binom{N}{M}$

\hspace{0.1in} - Repeat $2^k$ times and output the the set of centers $C$ that give least cost

\hspace{0.3in} - Make a call to {\bf Sample-centers$(P, k, \xi, 0, \{\})$} and select $C$ from the set of solutions

\hspace{0.4in} that gives the least cost.

{\bf Sample-centers$(P, k, \xi, i, C)$}

\hspace{0.1in} (1) If $(i = k)$ then add $C$ to the set of solutions

\hspace{0.1in} (2) else

\hspace{0.3in} (a) Sample a multiset $S$ of $N$ points with $D^2$-sampling (w.r.t. centers $C$)

\hspace{0.3in} (b) For all $s_i \in \{1, ..., \kappa\}$

\hspace{0.5in} (i) Let $T$ be the ${s_i}^{th}$ subset\footnote{To be able to define this, we fix an arbitrary ordering of points in $P$ that defines an ordering on the points in $S$. 
Also, for any set of size $N$, we fix an arbitrary ordering of the subsets of size $M$ of this set. } of $S$. 
$C \leftarrow C \cup \{\Gamma(T)\}$.\footnote{$\Gamma(T)$ denote the centroid of the points in $T$ i.e., $\Gamma(T) = \frac{\sum_{t \in T} t}{|T|}$.}

\hspace{0.5in} (ii) {\bf Sample-centers$(P, k, \xi, i+1, C)$}
\end{boxedminipage}
\caption{{\bf Find-$k$-means($P, k, \eps$)} gives $(1+\eps)$-approximation for data sets $P \subseteq \mathcal{X}$, where the distance measure $D$ over $\mathcal{X}$ that satisfies $\alpha$-approximate triangle inequality, $\beta$-approximate symmetry, Centroid property, and $(f, \gamma, \delta)$-sampling property.}
\label{fig:k}
\end{Algorithm}
\end{center}

\begin{theorem}[Main Theorem]\label{thm:other}
Let $f:\mathbb{R} \times \mathbb{R} \rightarrow \mathbb{R}$. 
Let $\alpha \geq 1$, $0 < \beta \leq 1$, and $0 < \delta < 1/2$ be constants and let $0 < \eps \leq 1$. Let $\eta = \frac{2\alpha^2}{\beta^2}(1 + 1/\beta)$. Let $D$ be a distance measure over space $\mathcal{X}$ that $D$ follows:
(a) $\beta$-approximate symmetry property,
(b) $\alpha$-approximate triangle inequality,
(c) Centroid property, and
(d) $(f, \epsilon, \delta)$-sampling property.
{\bf Find-$k$-means$\left(P, k, \eps \right)$} runs in time $O\left(nd  \cdot 2^{\tilde{O}\left(k \cdot f \left(\frac{\eps}{2 \eta}, 0.2 \right)\right)}\right)$ and gives a $(1 + \epsilon)$-approximation to the $k$-median objective for any point set $P \subseteq \mathcal{X}, |P| = n$.
\end{theorem}

We develop some notation first. 
Let $\C{i}$ be the set $C$ at the end of the $i^{th}$ recursive call of {\bf Sample-centers}.
To begin with $\C{0}$ is empty. 
Let $\itS{i}$ be the multiset $S$ sampled during the $i^{th}$ recursive call, and $\itT{i}$ be the corresponding set $T$ (which is the $s_i^{th}$ subset of $\itS{i}$).
Let $O_1, \ldots, O_k$ be the optimal clusters, and $c_1,...,c_k$ denote the respective optimal cluster centers. 
Further, let $m_i$ denote $|O_i|$.
Let $r_i$ denote the average cost paid by a point in $O_i$, i.e.,
$r_i = \frac{\sum_{p \in O_i} D(p,c_i)}{m_i}$.

We give an outline of the proof. 
Suppose before the  $i^{th}$ recursive call to {\bf Sample-centers}, 
we have found centers which are close to the centers of some $(i-1)$ clusters in the optimal solution. 
Conditioned on this fact, we show that in the next call, we are likely to sample enough {\em far-away} points from one of the remaining clusters.
We will show that the points that are closer to the currently chosen centers have a very small overhead if assigned to their closest center.
Further, we show that the samples from this new cluster are close to uniform distribution (c.f. Lemma~\ref{lem:key-repeat}). 
Since such a sample does not come from exactly uniform distribution, we cannot use the $(f, \gamma, \delta)$-sampling property directly. 

We now show that the following invariant will hold for all calls: let $\C{i-1}$ consist of centers
$c_1', \ldots, c_{i-1}'$ (added in this order). 
Then, with probability at least $\frac{1}{2^{i}}$, there exist distinct indices $j_1, \ldots, j_{i-1}$ such that for all $l = 1, \ldots, i-1$, $\Delta(O_{j_l},c_l') \leq (1 + \eps/\eta) \cdot \Delta(O_{j_l},c_{j_l})$.
Here, $\eta$ is a fixed constant that depends on $\alpha$ and $\beta$. With foresight, we fix the value of $\eta = \frac{2 \alpha^2}{\beta^2} \cdot (1 + 1/\beta)$.
Suppose this invariant holds for $\C{i-1}$ (the base case is easy since $\C{0}$ is empty).
We will then show that this invariant holds for $\C{i}$ as well (unless the current set of centers already give a $(1+\eps)$-approximation in which case there is nothing more to be shown). 
In other words, we just show that in the $i^{th}$ recursive call, with probability at least $1/2$,  the algorithm finds a center $c_{i}'$  such that
$ \Delta(O_{j_i},c_i') \leq (1 + \eps/\eta) \cdot \Delta(O_{j_i},c_{j_i})$, where $j_i$ is an index distinct from
$\{j_1, \ldots, j_{i-1}\}$.
This will basically show that at the end of the last call, we will have $k$ centers that give a $(1 + \eps)$-approximation with probability at least $2^{- k}$.

We now show that the invariant holds for $\C{i}$. We use the notation developed above for $\C{i-1}$. 
Let $J$ denote the set of indices $\{j_1, \ldots, j_{i-1}\}$. 
Now let $j_i$ be the index $j \notin J$ for which $\Delta(O_j, \C{i-1})$
is maximum. 
Intuitively, conditioned on sampling from $\cup_{l \notin J}O_{l}$ using $D^2$-sampling, it is  likely that
enough points from $O_{j_i}$ will be sampled.
A simple corollary of the next lemma shows that there is good chance that elements from the sets $O_{j_i}$ will be sampled, unless the current set of centers already gives the required $(1+\eps)$-approximation factor in which case we are done.

\begin{lemma}
\label{lem:sampled-repeat}
If $\frac{\sum_{l \notin J} \Delta(O_l, \C{i-1})}{\sum_{l=1}^k \Delta(O_l, \C{i-1})} \leq \eps/2$, then $\Delta(P, C^{(i-1)}) \leq (1 + \eps) \cdot \Delta_k(P)$.
\end{lemma}
\begin{proof}
We have:
{
\allowdisplaybreaks
\begin{eqnarray*}
\Delta(P, \C{i-1}) &=& \sum_{l \in J}  \Delta(O_l, \C{i-1}) +  \sum_{l \notin J} \Delta(O_l, \C{i-1})  \\
&\leq&  \sum_{l \in J} \Delta(O_l, \C{i-1})  +  \frac{\eps/2}{1 - \eps/2} \cdot \sum_{l \in J}  \Delta(O_l, \C{i-1})
\quad \textrm{(by hypothesis)}\\
&=& \frac{1}{1 - \eps/2} \cdot \sum_{l \in J} \Delta(O_l, \C{i-1}) \\
&\leq& \frac{1 + \eps/\eta}{1 - \eps/2} \cdot \sum_{l \in J}  \Delta_1(O_l) \quad \textrm{(using the invariant for $\C{i-1}$)}\\
&\leq& (1+\eps) \cdot  \sum_{l \in J}  \Delta_1(O_l) \\
&& \textrm{(using  $\eta = \frac{2\alpha^2}{\beta^2} \cdot (1 + 1/\beta) \geq 4$ since $\alpha \geq 1$ and $0 < \beta \leq 1$)}\\
&\leq& (1+\eps) \cdot  \sum_{l \in [k]}  \Delta_1(O_l) \qedhere
\end{eqnarray*}
}
\end{proof}

\noindent
We get the following corollary easily.
\begin{corollary}
\label{cor:sample-repeat}
If $\Delta(P, C^{(i-1)}) > (1+\eps) \cdot \Delta_k(P)$, then $ \frac{\Delta(O_{j_i}, \C{i-1})}{\sum_{l=1}^k  \Delta(O_l, \C{i-1})} > \frac{\eps}{2k}.$
\end{corollary}

The above lemma and its corollary say that: given that the current set of centers do not give the desired approximation, points in the set $O_{j_i}$ will be sampled with probability at least $\frac{\eps}{2k}$.
However, the points within $O_{j_i}$ are not sampled uniformly.
Some points in $O_{j_i}$ might be sampled with higher probability than other points.
Let us partition the points in $O_{j_i}$ into two sets, one with points that have a large conditional probability of being sampled and those that have small conditional probability.
Let $L_{j_i} \subseteq O_{j_i}$ such that 
\begin{equation}\label{eqn:partition}
\forall p \in L_{j_i}, \frac{D(p, C^{(i-1)})}{\Delta(O_{j_i}, C^{(i-1)})} \leq \frac{\beta \eps}{8 \alpha \eta} \cdot \frac{1}{m_{j_i}}. 
\end{equation}
Let $H_{j_i} = O_{j_i} \setminus L_{j_i}$.
In the next two lemmas, we show that for each point $p \in L_{j_i}, D(p, C^{(i-1)})$ is small and hence the contribution to the total potential with respect to $C^{(i-1)}$ is small. 

\begin{lemma}
\label{lem:key-repeat}
For any $l \notin  J$ and any point $p \in L_l$, $D(p, \C{i-1}) \leq \frac{\eps}{4 \eta} \cdot (r_l + D(p, c_l))$.
\end{lemma}
\begin{proof}
Using (\ref{eqn:partition}), we get the following:
\begin{eqnarray*}
\frac{\beta \eps}{8 \alpha \eta} \cdot \frac{1}{m_l} 
&\geq& \frac{e(p, C^{(i-1)})^2}{\Delta(O_l, C^{(i-1)})}
\geq \frac{D(p, c_t')}{\Delta(O_l, c_t')} \\
\Rightarrow \frac{\beta \eps}{8 \alpha \eta} \cdot \frac{1}{m_l}  
&\geq& \frac{D(p, c_t')}{\Delta(O_l, c_l) + m_l \cdot D(c_l, c_t')} \quad \textrm{(using the Centroid property)} \\
\Rightarrow \frac{\beta \eps}{8 \alpha \eta} \cdot \frac{1}{m_l}  
&\geq& \frac{D(p, c_t')}{m_l \cdot r_l + m_l \cdot D(c_l, c_t')} \\
\Rightarrow \frac{\beta \eps}{8 \alpha \eta} \cdot \frac{1}{m_l}  
&\geq& \frac{D(p, c_t')}{m_l \cdot r_l + \alpha \cdot m_l \cdot (D(c_l, p) + D(p, c_t'))}
\end{eqnarray*}
The last inequality is using the $\alpha$-approximate triangle inequality. 
We get the following from the above:
\begin{eqnarray*}
D(p, c_t') &\leq& \frac{\frac{\beta \eps}{8 \alpha \eta}}{1-\frac{\beta \eps}{8 \eta}} \cdot r_l + \frac{\frac{\beta \eps}{8\eta}}{1-\frac{\beta \eps}{8 \eta}} \cdot D(c_l, p) 
\leq \frac{\beta \eps}{4 \eta} \cdot (r_l + D(c_l, p)) \nonumber \\
\Rightarrow D(p, c_t') &\leq& \frac{\beta \eps}{4 \eta} \cdot (r_l + (1/\beta) \cdot D(p, c_l)) \quad \textrm{(using approximate symmetry)} \nonumber \\
\Rightarrow D(p, c_t') &\leq& \frac{\eps}{4 \eta} \cdot (r_l + D(p, c_l)) \qedhere
\end{eqnarray*}
\end{proof}

\begin{lemma}
For any $l \notin J$, $\Delta(L_l, C^{(i-1)}) \leq \frac{\eps}{2 \eta} \cdot \Delta_1(O_l)$.
\end{lemma}
\begin{proof}
Using the previous lemma, we get the following:
\begin{eqnarray*}
\Delta(L_l, C^{(i-1)}) 
&=& \sum_{p \in L_l} D(p, C^{(i-1)}) \\
\Rightarrow \Delta(L_l, C^{(i-1)})  
&\leq& \sum_{p \in L_l} \frac{\eps}{4\eta} \cdot (r_l + D(p, c_l)) \quad \textrm{(using Centroid property)} \\
\Rightarrow \Delta(L_l, C^{(i-1)})   
&\leq& \sum_{p \in O_l} \frac{\eps}{4 \eta} \cdot (r_l + D(p, c_l))
= \frac{\eps}{4 \eta} \cdot \sum_{p \in O_l}  r_l + \frac{\eps}{4 \eta} \cdot  \sum_{p \in O_l} e(p, c_l)^2 \\
\Rightarrow \Delta(L_l, C^{(i-1)})    
&\leq& \frac{\eps}{4 \eta} \cdot (m_l \cdot r_l + \Delta_1(O_l))
= \frac{\eps}{2 \eta} \cdot \Delta_1(O_l) \qedhere
\end{eqnarray*}
\end{proof}

Now, suppose that in the $i^{th}$ iteration, we manage to pick a point $c_i'$ such that 
$\Delta(H_{j_i}, c_i') \leq (1 + \eps/2\eta) \cdot \Delta_1(O_{j_i})$.
Then using the previous two lemmas, we have $\Delta(O_{j_i}, C^{(i-1)} \cup \{c_i'\}) \leq \Delta(H_{j_i}, c_i') + \Delta(L_{j_i}, C^{(i-1)}) \leq (1 + \eps/2\eta) \cdot \Delta_1(O_{j_i}) + (\eps/2\eta) \cdot \Delta_1(O_{j_i}) \leq (1 + \eps/\eta)\cdot \Delta_1(O_{j_i})$. 
In the next two lemmas, we show that there is a good probability of finding such a point.
The proof of these lemmas are almost the same (except for minor change in parameters) as the proof of Lemmas 12 and 13 in \cite{jks12}.

\begin{lemma}
\label{lem:clinaba-repeat}
Let $Q$ be a set of $n$ points, and $\gamma$ be a parameter, $0 < \gamma < 1$. Define a random variable $X$ as follows :
with probability $\gamma$, it picks an element of $Q$ uniformly at random, and with  probability $1-\gamma$, it does not
pick any element (i.e., is null). 
Let $X_1, \ldots, X_\ell$ be $\ell$ independent copies of $X$, where $\ell = \frac{4}{\gamma} \cdot f \left(\frac{\eps}{2 \eta}, 0.2\right).$
Let $T$ denote the multi-set of elements of $Q$ picked by $X_1, \ldots, X_\ell$. Then, with probability at least $3/4$,
$T$ contains a subset $U$ of size $f \left(\frac{\eps}{2 \eta}, 0.2\right)$ which satisfies
\begin{equation}\label{eq:clinaba-repeat}
\Delta(P, \Gamma(U)) \leq \left(1 + \frac{\eps}{2 \eta} \right) \cdot \Delta_1(P).
\end{equation}
\end{lemma}
\begin{proof}
Define a random variable $I$, which is a subset of the index set $\{1, \ldots, \ell\}$, as follows
$I = \{ t : X_t \mbox{ picks an element of $Q$, i.e., it is not null} \}$. 
Conditioned on $I = \{t_1, \ldots, t_r\}$, note that the random variables $X_{t_1}, \ldots, X_{t_r}$ are independent uniform samples from $Q$. 
Thus if $|I| \geq f \left(\frac{\eps}{2 \eta}, 0.2 \right)$, then sampling property wrt. $D$ implies that with probability at least $0.8$, the desired event~(\ref{eq:clinaba-repeat}) happens. 
But the expected value of $|I|$ is $4 \cdot f \left(\frac{\eps}{2 \eta}, 0.2 \right)$, and so, from Chernoff bound (Theorem~\ref{thm:chernoff}) $|I| \geq f \left(\frac{\eps}{2 \eta}, 0.2 \right)$ with probability at least $0.99$. 
Hence, the statement in the lemma is true.
\end{proof}

\begin{lemma}
\label{lem:final-repeat}
With probability at least $3/4$, there exists a subset $\itT{i}$ of $\itS{i}$ of size at most $f \left(\frac{\eps}{2 \eta}, 0.2 \right)$ such that
$ \Delta(H_{j_i}, \Gamma(\itT{i})) \leq \left(1 + \frac{\eps}{2 \eta}\right) \cdot \Delta_1(H_{j_i})$.
\end{lemma}
\begin{proof}
Recall that $\itS{i}$ contains $N = \frac{64 \alpha \eta k}{\beta \eps^2} \cdot f \left(\frac{\eps}{2 \eta}, 0.2 \right)$ independent samples of $P$ (using $D^2$-sampling). 
We are interested in $\itS{i} \cap H_{j_i}$.
Let $Y_1, \ldots, Y_N$ be  $N$ independent random variables defined as follows : 
for any $t$, $1 \leq t \leq N$, $Y_t$ is obtained by sampling an element of $P$ using $D^2$-sampling with respect to $\C{i-1}$. 
If this sampled element is not in $H_{j_i}$, then it just discards it (i.e., $Y_t$ is null) otherwise $Y_t$ is assigned that element.
Let $\gamma$ denote $\frac{\beta \eps^2}{16 \alpha \eta k}$. 
Corollary~\ref{cor:sample-repeat} and Lemma~\ref{lem:key-repeat} imply that $Y_t$ is assigned a particular element of $O_{j_i}$ with probability at least $\frac{\gamma}{m_{j_i}}$. 
We would now like to apply Lemma~\ref{lem:clinaba-repeat} (observe that $N = \frac{4}{\gamma} \cdot f \left(\frac{\eps}{2 \eta}, 0.2 \right)$). 
We can do this by a simple coupling argument as follows.
For any point $p \in H_{j_i}$, let the probability of it being sampled using $D^2$ sampling be denoted by $\frac{\lambda(p)}{m_{j_i}}$.
Note that $\forall p \in H_{j_i}, \lambda(p) \geq \gamma$.
So, for a particular element $p \in H_{j_i}$, $Y_t$ is assigned $p$ probability $\frac{\lambda(p)}{m_{j_i}} \geq \frac{\gamma}{m_{j_i}}$.
One way of sampling a random variable $X_t$ as in Lemma~\ref{lem:clinaba-repeat} is as follows --  first sample using $Y_t$. 
If $Y_t$ is null, then $X_t$ is also null. 
Otherwise, suppose $Y_t$ is assigned an element $p$ of $H_{j_i}$. 
Then $X_t$ is equal to $p$ with probability $\frac{\gamma}{\lambda(p)}$, and null otherwise. 
It is easy to check that with probability $\gamma$, $X_t$ is a uniform sample from $H_{j_i}$, and null with probability $1-\gamma$. 
Now, observe that the set of elements of $H_{j_i}$ sampled by $Y_1, \ldots, Y_N$ is always a superset of $X_1, \ldots, X_N$. 
We can now use Lemma~\ref{lem:clinaba-repeat} to finish the proof.
\end{proof}

Thus, we will take the index $s_i$ in Step 2(b) as the index of the set $\itT{i}$ as guaranteed by the lemma above.
Finally, by repeating the entire process $2^k$ times, we make sure that we get a $(1+\eps)$-approximate solution
with high probability.
Note that the total running time of our algorithm is $O\left( nd \cdot 2^k \cdot 2^{\tilde{O}\left(k \cdot f \left(\frac{\eps}{2 \eta}, 0.2 \right)\right)} \right)$.

\section{The $k$-means Problem}\label{sec:k-means}

Note that in the $k$-means problem $\mathcal{X} = \mathbb{R}^d$ and the distance measure $D$ is square of the Euclidean distance. 
Since this distance measure satisfies all the properties used in the analysis, we get the following result:

\begin{theorem}
{\bf Find-$k$-means}$(P, k, \eps)$ runs in time $O \left(nd \cdot 2^{\tilde{O}(k/\eps)} \right)$ and gives a $(1 + \eps)$-approximation to the $k$-means problem.
\end{theorem}
\begin{proof}
The proof follows from the proof of Theorem~\ref{thm:other} and the fact that the squared Euclidean distance measure over $\mathbb{R}^d$ satisfies $1$-approximate symmetry (trivial), $2$-approximate triangle inequality (see Lemma~3 in~\cite{jks12}), Centroid property (see Lemma~2 in \cite{jks12}), and $(f, \gamma, \delta)$-sampling property for $f(\gamma, \delta) = 1/(\gamma \delta)$ (see Lemma~1 in~\cite{jks12}).
\qedhere
\end{proof}

\bibliography{paper}

\begin{thebibliography}{1}

\bibitem{ab10}
Marcel~R. Ackermann and Johannes Bl\"{o}mer.
\newblock Bregman clustering for separable instances.
\newblock In {\em Proceedings of the 12th Scandinavian conference on Algorithm
  Theory}, SWAT'10, pages 212--223, Berlin, Heidelberg, 2010. Springer-Verlag.

\bibitem{abs10}
Marcel~R. Ackermann, Johannes Bl\"{o}mer, and Christian Sohler.
\newblock Clustering for metric and nonmetric distance measures.
\newblock {\em ACM Trans. Algorithms}, 6:59:1--59:26, September 2010.

\bibitem{Chen09}
K.~Chen.
\newblock On coresets for k-median and k-means clustering in metric and
  euclidean spaces and their applications.
\newblock {\em SIAM Journal on Computing}, 39(3):923--947, 2009.

\bibitem{jks12}
Ragesh Jaiswal, Amit Kumar, and Sandeep Sen.
\newblock A simple $D^2$-sampling based {PTAS} for k-means and other clustering
  problems.
\newblock {\em Algorithmica}, DOI: 10.1007/s00453-013-9833-9, 2013.

\bibitem{KumarSS10}
Amit Kumar, Yogish Sabharwal, and Sandeep Sen.
\newblock Linear-time approximation schemes for clustering problems in any
  dimensions.
\newblock {\em J. ACM}, 57(2):5:1--5:32, February 2010.

\end{thebibliography}

\end{document}